\documentclass[12pt]{amsart}
\usepackage{amsmath,amssymb,amsthm,latexsym,graphicx,cite,color}
\newtheorem{theorem}{Theorem}
\newtheorem{lemma}[theorem]{Lemma}
\newtheorem{proposition}[theorem]{Proposition}
\newtheorem{remark}[theorem]{Remark}
\newtheorem{corollary}[theorem]{Corollary}
\newtheorem{conjecture}[theorem]{Conjecture}
\newtheorem{definition}[theorem]{Definition}
\newtheorem{example}[theorem]{Example}
\def\C{\mathbb{C}}
\def\cV{\mathcal{V}}

\def\G{\Gamma}
\def\g{\gamma}
\def\N{\mathbb{N}}
\def\Q{\mathbb{Q}}
\def\R{\mathbb{R}}
\def\T{\mathbb{T}}
\def\Z{\mathbb{Z}}
\newcommand{\DSG}{{\it DSG}}

\begin{document}

\title[Generic dispersion relations]{Generic properties of dispersion relations for discrete periodic operators}
\author{Ngoc Do}
\address{Department of Mathematics, Missouri State University, Springfield, 65897 Missouri, USA}
\email{ngocdo@missouristate.edu}
\author{Peter Kuchment}
\address{Department of Mathematics, Texas A\&M University, College Station, 77843-3368 Texas, USA}
\email{kuchment@math.tamu.edu}
\author{Frank Sottile}
\address{Department of Mathematics, Texas A\&M University, College Station, 77843-3368 Texas, USA}
\email{sottile@math.tamu.edu}
\thanks{The authors were supported by the NSF.  The first two by grant DMS-1517938 and the third  by DMS-1501370.}

\begin{abstract}
An old problem in mathematical physics deals with the structure of the dispersion relation of the Schr\"odinger operator $-\Delta+V(x)$ in $\R^n$ with periodic potential near the edges of the spectrum, i.e.\ near extrema of the dispersion relation. A well known and widely believed conjecture says that generically (with respect to perturbations of the periodic potential) the extrema are attained by a single branch of the dispersion relation, are isolated, and have non-degenerate Hessian (i.e., dispersion relations are graphs of Morse functions). The important notion of effective masses in solid state physics, as well as Liouville property, Green's function asymptotics, etc. hinge upon this property.

The progress in proving this conjecture has been slow. It is natural to try to look at discrete problems, where the dispersion relation is (in appropriate coordinates) an algebraic, rather than analytic, variety. Moreover, such models are often used for computation in solid state physics (the tight binding model). Alas, counterexamples exist even for Schr\"odinger operators on simple 2D-periodic two-atomic structures. showing that the genericity fails in some discrete situations.

We start establishing in a very general situation the following natural dichotomy: the non-degeneracy of extrema either fails or holds in the complement of a proper algebraic subset of the parameters. Thus, a random choice of a point in the parameter space gives the correct answer ``with probability one.''

Noticing that the known counterexample has only two free parameters, one can suspect that this might be too tight for genericity to hold. We thus consider the maximal $\Z^2$-periodic two-atomic nearest-cell interaction graph, which has nine edges per unit cell and the discrete ``Laplace-Beltrami'' operator on it, which has nine free parameters. We then use methods from computational and combinatorial algebraic geometry to prove the genericity conjecture for this graph. Since the proof is non-trivial and would be much harder for more general structures, we show three different approaches to the genericity, which might be suitable in various situations.

It is also proven in this case that adding more parameters indeed cannot destroy the genericity result. This allows us to list all ``bad'' periodic subgraphs of the one we consider and discover that in all these cases genericity fails for ``trivial'' reasons only.
\end{abstract}
\maketitle

\section{Introduction}
Consider a $\Z^n$-periodic self-adjoint elliptic operator $L$ in $\R^n$. The issue discussed below can be formulated and studied in a more general setting, but the reader can think of the Schr\"odinger operator $L=-\Delta+V(x)$ with a real $\Z^n$-periodic potential $V(x)$, which we assume to be sufficiently ``nice\footnote{Assuming the potential even being $C^\infty$ does not seem to make the problem we discuss any easier.},'' e.g., $L_\infty$.

\subsection{Dispersion relation and spectrum}
We recall some notions from the spectral theory of periodic operators and solid state physics (see, e.g. \cite{AM,KuchBAMS,KuchBook,RS,Skr}).

For $k\in\R^n$ (called \textbf{quasimomentum} in physics) let us define the \textbf{twisted Schr\"odinger operator} $L(k)$ to
be $L$ applied to functions $u(x)$ on $\R^n$ that are \textbf{$k$-automorphic} (also called \textbf{Floquet}, sometimes \textbf{Bloch functions}, with quasimomentum $k$), i.e.
\begin{equation}\label{E:cyclic_k}
u(x+\gamma)=e^{ik\cdot\gamma}u(x) \mbox{ for any }\gamma\in\Z^n,
\end{equation}
where $k\cdot\gamma=\sum_j k_j\gamma_j$. In other words, $u(x)=e^{ik\cdot x}p(x)$, where the function $p(x)$ is $\Z^n$-periodic.

Then $L(k)$ is an elliptic operator in a line bundle over the torus\footnote{In another \textbf{Floquet multiplier} $z$ incarnation, see the Definition \ref{D:mult} below, $L(z)$ acts in the trivial bundle over the torus, but the operator depends analytically on $z$}.

The quasimomentum $k$ is well defined up to shifts by vectors from the lattice $2\pi\Z^n$ (the \textbf{dual lattice} $G^*$ to $G:=\Z^n$, i.e. consisting of all vectors $k$ such that $k\cdot\gamma\in 2\pi\Z$ for any $\gamma\in G$). Thus, it is sufficient to restrict $k$ to the \textbf{Brillouin zone} $B=[-\pi,\pi)^n$.

It is often convenient to factor out the $G^*$-periodicity and consider instead of vectors $k\in\C^n$ the complex vectors $Z$
with non-zero components
\begin{equation}\label{E:k_z}
z:=e^{ik}:=(e^{ik_1},\dots,e^{ik_n})\in\left(\C\backslash\{0\}\right)^n.
\end{equation}

\begin{definition}\label{D:mult}
Vectors $z$ in~\eqref{E:k_z} are called {\bf Floquet multipliers} (the name comes from the Floquet theory for ODEs \cite{MagnusWinkl_hills,Iakubovich_periodic,CoddLevinson}, where $n=1$).
\end{definition}

When the quasimomentum $k$ is real, the corresponding Floquet multiplier belongs to the \textbf{unit torus}
\begin{equation}
\T^n:=\{z\in\C^n\,|\,|z_j|=1,j=1,...,n\}\subset\C^n.
\end{equation}
In terms of Floquet multipliers $z$, the Floquet functions satisfy
\begin{equation}\label{E:cyclic_z}
u(x+\gamma)=z^\gamma u(x) \mbox{ for any }\gamma\in\Z^n,
\end{equation}
where $z^\gamma=e^{ik\cdot\gamma}=e^{i\sum_j k_j\gamma_j}$. In other words, $u(x)=z^x p(x)$, where the function $p(x)$ is $\Z^n$-periodic.

%
%
%
%

\subsection{Floquet decomposition}
The above discussion  suggests the use of Fourier series.
A standard argument justifies the decomposition of the (unbounded self-adjoint) operator $L$ in $L_2(\R^n)$ into the \textbf{direct integral} (see \cite{KuchBAMS,KuchBook,RS,Skr})
\begin{equation}\label{E:decomp}
L=\int\limits^{\bigoplus}_B L(k) dk .
\end{equation}
As $L(k)$ is an elliptic operator in sections of a line bundle over the torus $\R^n/\Z^n$, it has discrete spectrum
$\sigma(k):=\sigma(L(k))$ that consists of infinitely many  eigenvalues, each of finite multiplicity,
$$
\lambda_1(k)< \lambda_2(k)\leq \lambda_3(k)\leq \dots \to \infty\,.
$$
See \cite{KuchBook,KuchBAMS,RS,Skr} for more details.

%
%
\begin{definition}\label{D:dispers}
  The (real) \textbf{dispersion relation} (or  \textbf{Bloch variety}) $B_L$ of the periodic operator $L$ is the
  subset of\/ $\R^n_k\times \R_\lambda$, where $(k,\lambda) \in B_L$ if and only if
\begin{equation}\label{E:dispers}
   Lu = \lambda u \mbox{ has a non-zero  solution }u(x)=e^{ik\cdot x}p(x),
 \end{equation}
 where $p(x)$ is periodic with the same group of periods as the operator.

 The $j$th eigenvalue function $\lambda_j(k)$ is the \textbf{$j$th band function}.
\end{definition}

Thus, the dispersion relation is the graph of the multiple-valued function $k\mapsto \sigma(L(k))$.
\begin{figure}[ht!]
  \centering
  \includegraphics[scale=0.5]{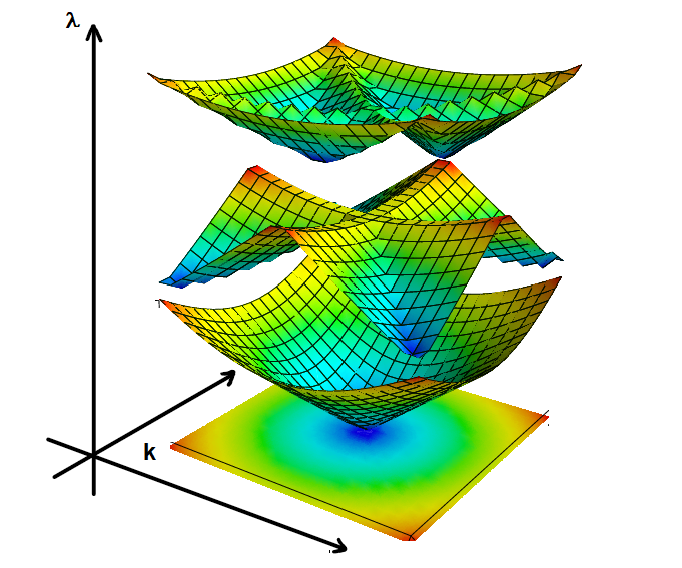}
  \caption{A dispersion relation. Three lower branches (band functions) are shown.}\label{F:dispers}
\end{figure}

\begin{remark}\label{R:complex}
  By allowing both the quasimomentum $k$ and the spectral parameter $\lambda$ to be complex, one defines the
  \textbf{complex Bloch variety $B_{L,\C}$}. In the complex domain, though, numbering the eigenvalues in their order becomes impossible. In many cases, they become branches of the same irreducible analytic function.
\end{remark}

The following properties of the dispersion relation (Bloch variety) are well known:

\begin{proposition}[\cite{KuchBAMS,KuchBook}]\label{P:analyt}
\indent
\begin{enumerate}
\item The complex Bloch variety is an analytic subvariety of $\C^n_k\times\C_\lambda$. Namely, it is the set of all zeros of an entire function $f(k,\lambda)$ of a finite exponential order\footnote{In some instances, the exponential estimate becomes important, see Section \ref{S:remarks}.} on $\C^n_k\times\C_\lambda$.
\item The projection of the real Bloch variety onto the real $\lambda$-axis is the spectrum $\sigma(L)$ of the operator $L$ in $\R^n$.
\item The projection of the graph of the $j$th band function into the real $\lambda$-axis is a finite closed interval called the \textbf{$j$th spectral band $I_j$}. The spectral bands might overlap, or leave open spaces in between called \textbf{spectral gaps}, see Fig.~\ref{F:spectrum}.
\end{enumerate}
\end{proposition}

%
%
\begin{figure}[ht!]
  \centering
  \begin{picture}(282,82)
      \put(-1,-0.5){\includegraphics{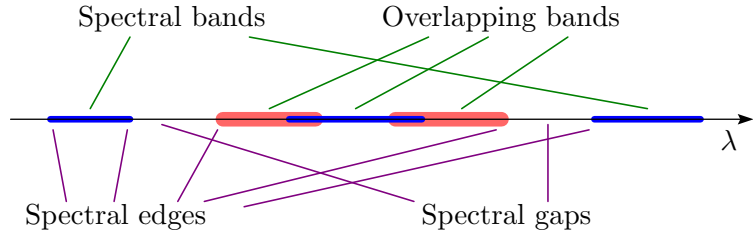}}
    \put( 5,0){\small Spectral edges}
    \put(155,0){\small Spectral gaps}

    \put( 25,75){\small Spectral bands}
    \put(140,75){\small Overlapping bands}

    \put(268,28){\small$\lambda$}
  \end{picture}
  \caption{Spectral bands and gaps.}\label{F:spectrum}
\end{figure}

\subsection{``Extended'' dispersion relations}

Besides varying the quasimomentum $k$, one might vary other parameters of the operator, e.g. of the periodic potential $V\in L_\infty(W)$, or introduce a periodic diffusion coefficient (metric) in the operator: $-\nabla\cdot D(x)\nabla u+ Vu$ (assuming that ellipticity is preserved). One can thus consider an extended dispersion relation in an appropriate Banach space
of quadruples $(D,V,k,\lambda)$. As long as ellipticity is preserved, an analog of Proposition \ref{P:analyt} still holds \cite{KuchBAMS}.

\subsection{The spectral edge conjecture}

\textbf{Spectral edges} are the endpoints of spectral gaps (see Fig.~\ref{F:spectrum}). By Proposition \ref{P:analyt}, they correspond to some of the extremal values of band functions. We are interested in the generic (with respect to perturbation of the periodic potential or other parameters of the operator) structure of these extrema. The genericity can be understood in a variety of ways, e.g. holding for a second Baire category set of potentials in an appropriate (Banach) space of potentials (the most likely situation), or stronger one - for a dense open subset, or even stronger - in the exterior of an analytic (or even algebraic) subset in the space.

An old conjecture, more or less explicitly formulated in a variety of sources, e.g. in \cite[Conjecture 5.25]{KuchBAMS}
or in~\cite{KuchBook,Nov,Nov2,Colin}, deals with the structure of the dispersion relation of the Schr\"odinger operator in
$\R^n$ with periodic potential near the edges of the spectrum, i.e. near (some of) the extrema of the dispersion relation. This well known
and widely believed conjecture says that generically the band functions are Morse functions. We make this precise below:

\begin{conjecture}\label{Con:gener}
Generically (with respect to the potentials and other free parameters
of the operator, e.g.\ metric in the Laplace-Beltrami operator), the extrema of band functions satisfy the following
conditions:
\begin{enumerate}
\item Each extremal value is attained by a single band $\lambda_j(k)$.
\item The loci of extrema are isolated.
\item The extrema are non-degenerate, i.e. at them the corresponding band functions have non-degenerate Hessians.
\end{enumerate}
\end{conjecture}
Notice that (3) implies (2).

This conjecture asserts that generically near a spectral edge the dispersion
relation has a parabolic shape and thus resembles the dispersion
relation at the bottom of the spectrum of the free operator $-\Delta$. This in turn would trigger appearance of various properties analogous to those of the Laplace operator. One can mention, for instance, electron's effective masses in solid state theory \cite{AM,Kittel}, Green's function asymptotics \cite{Bab,Kha,KhaKucRai,KucRai}, homogenization \cite{BirSus,BensHom}, Liouville type theorems \cite{KucPin,KucPin2,AvelLin,Mos,KhaKuc}, Anderson localization  \cite{AizMol},  perturbation of discrete spectra in gaps in general \cite{Bir1,Bir2,Bir3}, and others.

The intuition for formulating this conjecture is explained in \cite{KuchBAMS}.
The progress in proving it has been very slow. We summarize here briefly the successes achieved so far. It is known, that the expected parabolic structure always (not only generically) holds at the bottom of the spectrum of Schr\"odinger operator with a periodic electric potential \cite{KirschSim} (which is not necessarily true if a magnetic field is involved \cite{Ster}). In \cite{KR}, the statement (1) of the conjecture was proven. The full conjecture was proven in \cite{Colin} in $2D$ for any given number of bands and small smooth potentials. The statement (2) was proven in $2D$ \cite{FK} in a stronger form, even without the genericity clause.

\subsection{Discrete version}

It is natural to look first at discrete problems (arising, for instance, as tight binding
approximations in solid state physics~\cite{AM,Kittel}).
Then, the dispersion relation becomes (in the Floquet multipliers coordinates) an algebraic variety~\cite{Gis,KuchDisc}.
Alas, counterexamples to generic non-degeneracy in the discrete \cite{FK}, as well as quantum graph \cite{FK,Berk} case have been known, even for simple two-atomic $2D$ structures~\cite{FK}.

Our first result is that in the discrete periodic (in fact, in a much more general ``algebraically fibered'') case, the following dichotomy holds: either the set of parameters for which there are degenerate critical points is contained in a proper algebraic subset, or it contains the complement of a such subset). Thus, testing a ``random'' sample of parameters should provide an ``almost surely correct'' answer. Other main results are described below.

\subsection{The structure of the paper:}
In Section \ref{S:discrete} we provide a detailed description of the discrete case. The content of this section is well known, but we still provide some definitions and state basic facts to make the text more self-contained.
The result (Theorem \ref{T:dichotomy}) on the dichotomy and its proof are presented in Section \ref{S:main}.
In Section \ref{S:example}, we consider a rather large class of $\Z^2$-periodic graphs, those that are two-atomic, with no multiple edges or loops, and in some sense only with local interactions (i.e, atoms in a cell interact only with atoms in the cells having a common edge with the latter one) and study the maximal graph in this class. Since \cite{FK} provides an example of such graph with a Schr\"odinger operator, for which generic non-degeneracy does not hold and the space of free parameters $\alpha$ is only two-dimensional, we look at a ``divergence type'' operator, which has nine degrees of freedom and thus, as we explain there, better chance for Conjecture~\ref{Con:gener} to hold. And, indeed, we prove (Theorem \ref{T:example}) the conjecture for this graph. The general idea has always been clear: the joint locus of $n$ (in our case, four) polynomials ``should'' have co-dimension $n$. If this were true, the life would have been easy proving the conjecture, and examples like the one in \cite{FK} would not exist. Regretfully, this is true only generically, and checking it in a particular case seems to be really hard. Moreover, as our work shows, there could be ``hidden'' misbehaving components, with smaller co-dimension, which luckily turn out to be irrelevant. Due to complexity of situation and with no clear path forward to more complex structures, we find it useful to present three approaches to establishing the result either rigorously, or ``with high confidence.''
These are based on using algebraic geometry analytic and computational tools for a numerical evidence, an ``almost surely'' verification, and then an actual proof that relies on an exact (but highly non-trivial) count of solutions.
We then use symbolic computation to find all maximal substructures (i.e., with some edges dropped) of the class of graphs studied, for which Conjecture 7 does not hold. An interesting observation is that all these cases have only trivial reasons for degeneracy (obvious multiplicity of some branches of the dispersion relation). Thus, we obtain the answers for the whole class of graphs in question.different

Final remarks can be found in Section \ref{S:remarks} and acknowledgments in Section \ref{S:thanks}.

\section{Description of the discrete case}\label{S:discrete}

We consider a discrete situation, i.e. when the group $G=\Z^n$ acts on a graph $\Gamma$ with the set of vertices $V$ and edges $E$. We write $x\sim y$  for two vertices $x,y\in V$ when there exists an edge connecting them\footnote{It is easy to modify the notions and proofs below for the case when multiple edges are allowed between a pair of vertices.}.

Let us start with making this notion precise. These details, which and more can be found in \cite{BerKuc}, are provided to make the text more self-contained.

\subsection{Periodic graphs}\label{S:periodic}

\begin{definition}\label{D:periodic}
An infinite graph $\G$ is said to be  \textbf{periodic} (or $\Z^n$-periodic) if $\G$ is equipped
with an action of the free abelian group $G=\Z^n$, i.e.   a mapping $(g,x)\in G\times \G \mapsto gx\in\G$, such that
the following properties are satisfied:
\begin{enumerate}
\item {\bf Group action:}\\ For any $g\in G$, the mapping $x\mapsto gx$ is a bijection of $\G$;\\ $0x=x$ for any $x\in\G$, where $0\in G=\Z^n$ is the neutral element;\\ $(g_1g_2)x=g_1(g_2x)$ for any $g_1,g_2\in G, x\in\G$.

\item {\bf Faithful:} If $gx=x$ for some $x\in\G$, then $g=0$.
\item {\bf Discrete:} For any $x\in\G$, there is a neighborhood $U$ of $x$ such that $gx\not\in U$ for $g\neq 0$.
\item {\bf Co-compact:} The space of orbits $\G/G$ is finite. In other words, the whole graph can be obtained by the $G$-shifts of a finite subset.
\item {\bf Structure preservation:}
\begin{itemize}
\item $gu\sim gv$ if and only if $u\sim v$. In particular, $G$ acts bijectively on the set of edges.
\item If other parameters are present (e.g., weights at vertices or at edges), the action preserves their values.
\end{itemize}
 \end{enumerate}
\end{definition}

A simple way to visualize this is to think of a graph $\G$ embedded into $\R^n$ ($n\geq 3$) in such a way that it is invariant with
respect to the shifts by integer vectors $g\in\Z^n\subset\R^n$, which produces an action of $\Z^n$ on $\G$.

\begin{definition}\label{D:fundamental} Due to co-compactness (4), there exists a finite part $W$ of $\G$ such that
\begin{itemize}
\item The union of all $G$-shifts of $W$ covers $\G$,
    $$
    \bigcup\limits_{g\in G} gW=\G.
    $$
\item Different shifted copies of $W$, i.e. $g_1W$ and $g_2W$ with $g_1\neq g_2\in G$, do not share any vertices.
\end{itemize}
Such a  compact subset $W$ is called a {\bf fundamental domain} for the action of $G$ on $\G$.
\end{definition}
\begin{remark}\label{R:Wboundary}
Note that a fundamental domain $W$ is not uniquely defined.
\end{remark}

An example of a fundamental domain is shown in Fig. \ref{F:hexagon}.
\begin{figure}
\begin{center}
\includegraphics[scale=0.5]{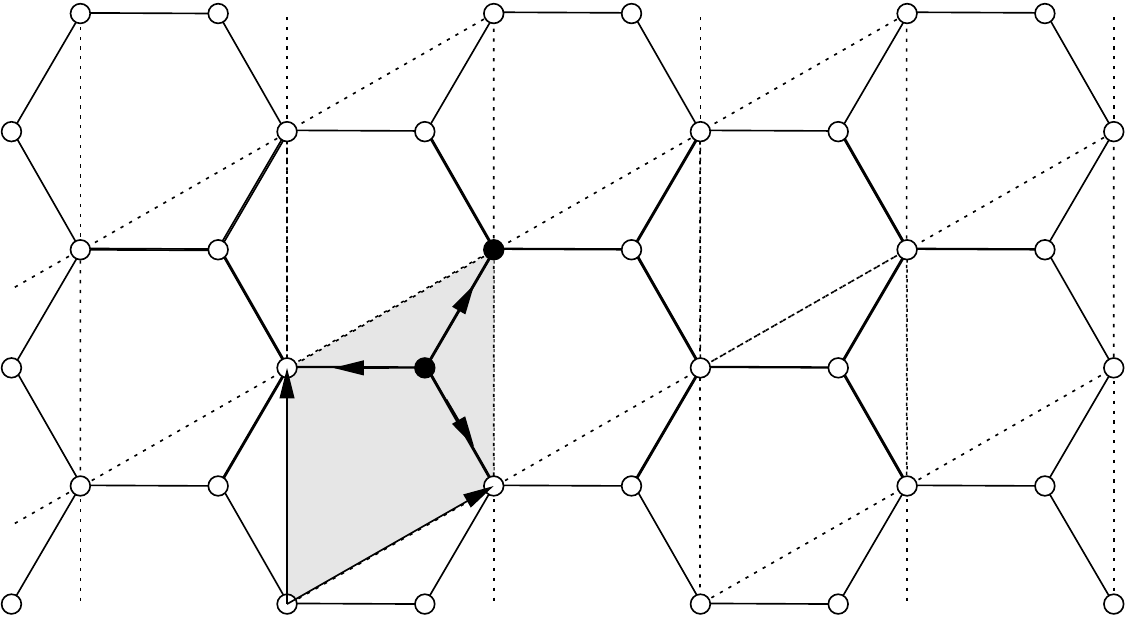}
\end{center}
\caption{A hexagonal ``graphene'' lattice as a graph $\Gamma$  with a two-atom fundamental domain shaded.}\label{F:hexagon}
\end{figure}
%
%

\subsection{Floquet-Bloch theory}\label{S:Floquet}

\subsubsection{Floquet transform on periodic graphs}\label{SS:Floq_trans_comb}
As in the continuous case, the standard idea of harmonic analysis suggests that, as long as we are dealing with a linear problem that commutes with an action of the abelian group $G=\Z^n$, Fourier series expansion\footnote{I.e., expansion into irreducible representations.} with respect to this group should simplify the problem. Its implementation leads to what is known as the \textbf{Floquet transform}. Indeed, what one needs to do is to expand functions on the graph $\G$ into the unitary characters $\g_k$ or, equivalently, $\g_z$.

Let $\G$ be a $\Z^n$-periodic graph and $f$ be a finitely supported (or sufficiently fast decaying) function defined on the
set of vertices $\cV$ of $\G$.

\begin{definition}\label{D:Floq_trans}
We define \textbf{Floquet transform} of $f$ as
\begin{equation}\label{E:Floquet}
  f(v)\mapsto \hat{f}(v,z)=\sum\limits_{g\in\Z^n} f(gv)z^{-g},
\end{equation}
where $gv$ denotes the action of $g\in\Z^n$ on the vertex $v\in V$ and
$z=(z_1,\dotsc,z_n)\in (\C\backslash 0)^n$ is the  Floquet multiplier.

Using quasimomenta instead of Floquet multipliers,~\eqref{E:Floquet} becomes
\begin{equation}\label{E:Floquet-k}
  f(v)\mapsto \hat{f}(v,e^{ik})=\sum\limits_{g\in\Z^n} f(gv)e^{-ik\cdot g}.
\end{equation}
\end{definition}

The reader can notice that~\eqref{E:Floquet}--\eqref{E:Floquet-k} is just the Fourier transform with respect to the action
of $G=\Z^n$ on the set $V$ of vertices.

Basic properties of the Floquet transform can be found in \cite{BerKuc,KuchBAMS}.

\begin{remark}\label{R:chop}
One can also interpret the Floquet transform as follows: one takes a function $\phi$ on $\G$ and cuts $\G$ into non-overlapping pieces by restricting $f$ to the shifted copies $gW$ of a fundamental domain $W$. These pieces are shifted back to $W$ and then are taken as (vector valued) Fourier coefficients of the Fourier series~\eqref{E:Floquet-k} that defines the Floquet transform.
\end{remark}

\begin{definition}
Let $W$ be a (finite) fundamental domain of the action of the group $G=\Z^n$ on $\G$. We will denote
$\hat{f}(v,z)|_{v\in W}$ by $\hat{f}(z)$, where the latter expression is considered as a function of $z$ with values in the space of functions on $W$. In other words, $\hat{f}(z)$ takes values in $\C^{|W|}$.
\end{definition}


\subsubsection{Floquet transform of periodic difference operators}\label{SS:Floq_trans_diffop}
Let $A$ be a difference operator on a $\Z^n$-graph $\G$. In other words, $A$ is an infinite $|V|\times|V|$ matrix.
We assume that $A$ has finite order, meaning that in each row it has only finitely many non-zero entries (in other words, only
finitely many neighbors of each vertex are involved).
We also assume that $A$ is periodic, i.e.\ commuting with the action of the group $\Z^n$. After Floquet transform this operator becomes the operator of multiplication by a matrix $A(z)$ of size $|W|\times|W|$ depending rationally on the Floquet multiplier $z$ (or analytically on the quasimomentum $k$).

For the reader, who has not done these simple calculations for a periodic graph (apologies for other readers), we provide the example of the Laplace operator on the regular ``graphene'' $2D$ lattice $\G$ (see Fig. \ref{fig:hexagon}).
\begin{figure}[ht!]
  \begin{picture}(0,0)%
    \includegraphics{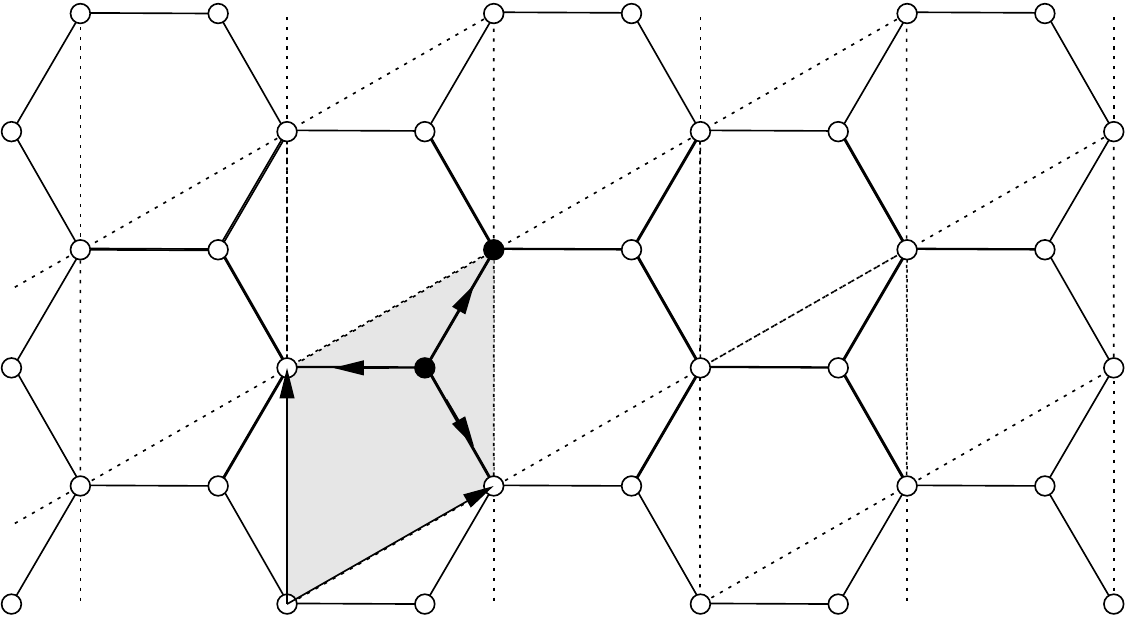}%
  \end{picture}%
  \setlength{\unitlength}{4144sp}%
  \begin{picture}(5146,2865)(83,-1973)
    \put(2386,-1501){$b''$}
    \put(2206,-601){$f$}
    \put(2161,-286){$b$}%
    \put(541,-1861){$\G$}%
    \put(1182,-826){$b'$}%
    \put(2136,-844){$a$}%
    \put(1533,-1604){$W$}%
    \put(2026,-1136){$h$}%
    \put(1726,-900){$g$}%
    \put(1751,-717){$\G$}%
    \put(1142,-1933){$b'''$}%
    \put(1216,-1321){$\vec e_2$}%
    \put(1891,-1726){$\vec e_1$}%
  \end{picture}%
  \caption{The hexagonal lattice $\G$ and a fundamental domain $W$
    together with its set of vertices $V(W)=\{a,b\}$ and
    set of edges $E(W)=\{f,g,h\}$.}
  \label{fig:hexagon}
\end{figure}
The group $\Z^2$ acts on $\G$ by the shifts by vectors
$p_1 e _1+p_2 e _2$, where $(p_1,p_2)\in\Z^2$ and vectors
$e_1=(3/2,\sqrt 3/2)$, $e_2=(0,\sqrt 3)$ are shown in
Figure \ref{fig:hexagon}. We choose as a fundamental domain (Wigner-Seitz cell) of this
action the shaded parallelogram region $W$.
Two black vertices $a$ and $b$ belong to $W$, while $b'$, $b''$, and $b'''$ lie in shifted copies of $W$.
Three edges $f,g,h$, directed as shown in the picture, belong to $W$.

We consider the Laplace operator
$$
Af(v)=\sum\limits_{w\sim v}f(w)-3f(v).
$$
One can find the ``symbol'' $A(z)$ (or $A(k)$ in terms of the quasimomenta) by applying $A$ to functions $f$ automorphic
with the character $z=e^{ik}$. This leads to
%
\begin{equation}
\begin{array}{ccc}
(Af)(a)&=&-3f(a)+(z_1^{-1}+z_2^{-1}+1)f(b)\vspace{1pt}\\
       &=&-3f(a)+(e^{-ik_1}+e^{-ik_2}+1)f(b),\ \mbox{and}\vspace{3pt}\\
(Af)(b)&=&(z_1+z_2+1)f(a)-3f(b)\vspace{1pt}\\
       &=&(e^{ik_1}+e^{ik_2}+1)f(a)-3f(b).
\end{array}
\end{equation}
We thus obtain the expression for the ``symbol'' of $A$:
\begin{equation}
\begin{array}{ccc}
A(z)&=&\left(
       \begin{array}{cc}
         -3 & z_1^{-1}+z_2^{-1}+1 \\
         z_1+z_2+1 & -3\\
       \end{array}
     \right) \vspace{5pt}\\
A(k)&=&\left(
       \begin{array}{cc}
         -3 & e^{-ik_1}+e^{-ik_2}+1 \\
         e^{ik_1}+e^{ik_2}+1 & -3 \\
       \end{array}
     \right)
\end{array}
\end{equation}
The matrix $A(z)$ depends rationally on $z$-variables.
In fact, it is a Laurent polynomial.
Since variables $z$ belong to the unit torus $\T^n$, no singularities of
$A(z)$ appear there.  It is thus possible to multiply the matrix $A(z)$ by $z_1^mz_2^m\dots
z_n^m$
with a sufficiently high power $m$, so the resulting matrix $\widetilde{A}(z)$ has
polynomial entries. E.g., in the above example, by multiplying by $z_1 z_2$ (i.e. $m=1$), the dispersion relation in the $(z, \lambda)$ space can be given as follows
\begin{equation}\label{E:dispdisc_set}
\{(z,\lambda)\,|\,\det\left(\widetilde{A}(z)-\lambda z_1z_2\right)=0\}, \text{ where }
\end{equation}
$$
\begin{array}{c}
\widetilde{A}(z)=\left(
       \begin{array}{cc}
         -3z_1z_2 & z_1+z_2+z_1z_2 \\
         z_1^2z_2+z_1z_2^2+z_1z_2 & -3z_1z_2\\
       \end{array}
     \right)
\end{array}
$$
The dispersion relation is  an algebraic variety of codimension 1 in $\C^3$.

An analogous construction holds for general finite order periodic
difference operators on graphs, with $z_1z_2$ being replaced by the product
$z_1^mz_2^m\dots z_n^m$.

We can consider the equation
\begin{equation}\label{E:dispdisc_Eq}
\Phi(z,\lambda):=\det\left(\widetilde{A}(z)-\lambda z_1^mz_2^m\dots z_n^m\right)=0
\end{equation}
as an implicit description of the graph of a multiple-valued function
$$
F: z\mapsto \lambda
$$ that
shows dependence of the spectrum of $\widetilde{A}$ on the parameters $z$.

The matrix $A$ depends polynomially on extra parameters $\alpha$ (e.g., weights at vertices and/or edges,
potentials, etc.),
\begin{equation}\label{E:dispdiscParam}
\Phi(\alpha,z,\lambda):=\det\left(\widetilde{A}(\alpha,z)-\lambda z_1^mz_2^m\dots z_n^m\right)=0.
\end{equation}
Correspondingly, we have a family of functions
$$
F_\alpha:=F(\alpha,\cdot)\colon z\mapsto \lambda.
$$
Thus, the question can be reformulated as follows:
\begin{center}
\textbf{Does non-degeneracy of all critical points of the function $F_\alpha$ on\\ the torus  hold generically with respect
  to the parameters $\alpha$?}
\end{center}
For the particular class of discrete graphs considered in the rest of the paper, we prove genericity in the strongest sense: as being valid outside of a proper algebraic subset. This is clearly not expected to happen for PDEs. Establishing in the latter situation any weaker type of genericity would be valuable.

\section{The critical point dichotomy}\label{S:main}

The matrix $A(\alpha,z,\lambda)$ introduced above, and thus
$\widetilde{A}(\alpha,z,\lambda)$ as well, has a very special structure, due to the periodicity. However, an important dichotomy holds in a very general
situation, without any connection to periodicity.

We now formulate and prove the following dichotomy statement.

\begin{theorem}\label{T:dichotomy}
 Let $U\subset\C^n$ be a neighborhood of the torus $\T^n$, $P(z)$ be a polynomial, and
$A(\alpha,z)$ be a finite size matrix polynomially dependent on the parameters $\alpha\in\C^m$
 and $z\in\C^n$.
 Consider for any $\alpha$ the equation
\begin{equation}\label{E:dispdisc_Eqn}
\Phi_\alpha(z,\lambda):=\det\left( A(\alpha,z)-\lambda P(z) I\right)=0
\end{equation}
as an implicit description of the graph of a multiple-valued function
$$
F_\alpha\colon z\mapsto \lambda,
$$
that shows dependence of the (weighted by $P(z)$) spectrum of $A$ on the parameters $z$.

Then the set $DV$ of points $\alpha$ where $F_\alpha$ has a degenerate critical point on (and near to) torus $\T^n$, either
belongs to a proper algebraic subset of $\C^m$ or it contains the complement of such a set.
\end{theorem}

\begin{proof}
Let us first define the set $DC\subset \C^m\times U\times\C$ of points
$(\alpha,z,\lambda)$, where one encounters a degenerate critical
point of $F_\alpha$.
Projecting into the $\alpha$-space $\C^m$, we obtain the set $DV$ of parameters
$\alpha$ that we are interested in.

This can be easily done in terms of the function $\Phi_\alpha(z,\lambda)$, by using the
condition that this function, as well as implicitly computed gradient and Hessian of
$\lambda$ with respect to $z$ all vanish. This, indeed, can be done by implicit
differentiation using the equation $\Phi_\alpha(z,\lambda)=0$:
\begin{equation}\label{E:equations}
\begin{cases}
\Phi_\alpha(z,\lambda)=0,\\
\frac{\partial \lambda}{\partial z_j}=0 \mbox{ for all }j=1,\dots,n,\\
\det\left(\frac{\partial^2\lambda}{\partial z_i\partial z_j}\right)=0.
\end{cases}
\end{equation}
The gradient and Hessian of $\lambda$ with respect to $z$ can be obtained by implicitly differentiating equation
$\Phi_\alpha(z,\lambda)=0$.
This produces rational expressions whose vanishing is equivalent to vanishing of their polynomial numerators.
Thus, one obtains the system of $n+2$
polynomial equations
\begin{equation}\label{E:system}
\begin{cases}
\Phi_\alpha(z,\lambda)=0,\\
P_j(\alpha,z,\lambda)=0,\ \mbox{ for } j=1,\dots,n,\\
H(\alpha,z,\lambda)=0,
\end{cases}
\end{equation}
which describes the set $DC$.
We ask: how large can the projection $DV:=\pi DC$ of $DC$ into the space $\C^m_\alpha$ of
parameters $\alpha$ be? As a projection of an algebraic set, its closure is algebraic.
If the dimension of the projection is less than $m$, then it is a proper subset of $\C^m$
and we have genericity of the parameters for which all critical points are nondegenerate.
The alternative is that the closure of $DV$ is $\C^m$, so that for generic $\alpha$ there are degenerate
critical points.
In this last case, there may yet be a proper algebraic set of parameters $\alpha$ for which all critical points are
nondegenerate.
\end{proof}

Our desire is to have the set $DV$ ``small,'' i.e. the first alternative of Theorem \ref{T:dichotomy} to take place.
However, as we have already
mentioned, even for ``two-atomic'' periodic discrete structures this is not necessarily the case~\cite{FK}. So, how can one tell in which of two options of the dichotomy we are in any particular situation? While we do not have a complete
answer to this, the following ``random test'' follows from Theorem~\ref{T:dichotomy}.
Let us pick a value of $\alpha$ \textbf{``randomly''} and compute the dispersion relation.
If it has no degenerate critical points, then we know ``with probability one'' that $DV$ is contained in a proper algebraic
subset, and thus non-degeneracy is generic.
If instead we determine that $\alpha\in DV$, then we know that ``with probability $1$'' degeneracy is generic.
Indeed, the chances for a randomly selected point to belong to a given proper algebraic subset are zilch.

\begin{corollary}\label{C:test}
 A random sample of parameters with respect to any absolutely continuous probability distribution gives the correct answer (generic non-degeneracy or not) with probability 1.
\end{corollary}

We also formulate the following conjecture:
\begin{conjecture}\label{C:upscaling}
Generic non-degeneracy survives under extending the set of parameters, e.g. if in addition to varying the potential, we start varying the metric as well.  In other words, changing more parameters cannot make the situation worse.
\end{conjecture}
We prove this for the class of graphs studied below.
\section{An example and three alternative approaches}\label{S:example}
One can ask whether one can avoid a random choice of parameters.
As we have implied at the end of the Introduction, the answer is ``yes,'' but it is not that easy to implement.
Namely, there are $n+2$ polynomial equations~\eqref{E:system} determining the set $DC$.
If the codimension of $DC$ were exactly $n+2$, then projecting onto the space $\C^m_\alpha$ along the $(n+1)$-dimensional
$U\times\C$ would produce a set $DV$ of at least codimension $1$ and thus for generic parameters $\alpha$, all critical
points would be nondegenerate. The end of the story! This dimension-counting was part of our intuition behind Conjecture~\ref{Con:gener}.
Unfortunately, the codimension of an algebraic set (or at least of some of its irreducible components) could be less than the number of the defining equations (in our case, $n+2$), as for instance the example of~\cite{FK} confirms. So, one can try to figure out the dimensions (and thus codimensions) of the irreducible components, which is, as the example below shows, sometimes possible, but far from being easy.

\subsection{An example of a discrete structure}
We consider the discrete periodic graph $\Gamma$ shown in Fig.~\ref{F:sample}.
The square fundamental domain $W$ contains two vertices (``atoms'') $a$ and $b$ and nine edges shown with solid lines.
We allow connections inside $W$ and to its four adjacent copies, introducing thus more free parameters, which hopefully would make the Conjecture \ref{Con:gener} more likely to hold, if Conjecture \ref{C:upscaling} is to be believed. No loops or multiple edges are allowed.
Shifts of $W$ by integer linear combinations of basis vectors $e_1$ and $e_2$ tile the plane.
We write $V$ and $E$ for the sets of vertices and edges of $\Gamma$ respectively.

This graph and its periodic sub-graphs describe all $2D$-periodic\footnote{Most of them non-planar).} two-atomic ``nearest cell interaction'' structures without loops and multiple edges.

The graph $\Gamma$ is equipped with a periodic \textbf{weight} function $\alpha$ (an analog of a metric, or an anisotropic
diffusion coefficient) that assigns to each edge a non-negative number.
\begin{figure}[ht!]
  \centering
  \begin{picture}(240,240)
    \put(-0.5,-0.5){\includegraphics{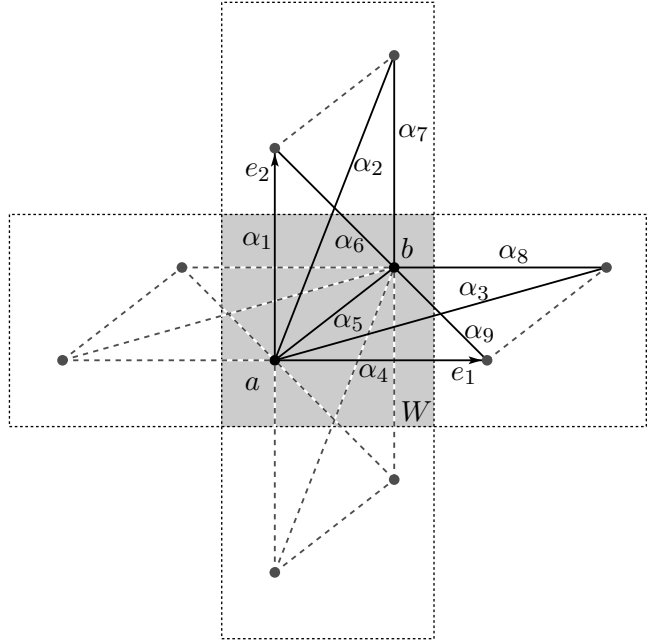}}
    \put(89,94){\small$a$}    \put(148,144){\small$b$}     \put(148,82.5){\small$W$}
    \put(167, 98){\small$e_1$} \put(89,175){\small$e_2$}
    \put( 88,149){\small$\alpha_1$}
    \put(130,177){\small$\alpha_2$}
    \put(170,130){\small$\alpha_3$}
    \put(132, 99){\small$\alpha_4$}
    \put(122.5,118){\small$\alpha_5$}
    \put(123,148){\small$\alpha_6$}
    \put(147,190){\small$\alpha_7$}
    \put(185,144){\small$\alpha_8$}
    \put(172,114){\small$\alpha_9$}
  \end{picture}
  \caption{The basic period vectors are $e_1$ and $e_2$.
    The fundamental domain $W$ contains two atoms $a$ and $b$ and nine (solid) edges. The dotted edges and other atoms are
    obtained by shifting the fundamental domain by integer linear combinations of $e_1, e_2$. Numbers $\alpha_j$ are
    weights associated with the solid edges.}\label{F:sample}
\end{figure}

\begin{remark}\label{R:moreparam}
Notice that the graph is diatomic, like in \cite{FK}, but the freedom of choosing parameters is nine-dimensional, while it was only two-dimensional in \cite{FK}. The intuition (Conjecture \ref{C:upscaling}) is that this should help the genericity.
\end{remark}

Let us denote the set of non-negative real numbers by $\R_+$.
Given $\alpha=(\alpha_1,\ldots, \alpha_9)\in \R_+^9$, we can assign the weights $\alpha_j, j=1,\dotsc,9$, to the edges
from the fundamental domain $W$ as shown in Fig.~\ref{F:sample}.
The entire structure $\G$ and all edge weights can be obtained from $W$ by $\mathbb{Z}^2$-shifts (with the basis $e_1,e_2$).
Define a divergence\footnote{This is a discrete analog of a second order divergence type elliptic partial differential operator.} (or Laplace-Beltrami) type operator $L_\alpha$ acting on the graph $\G$ as follows:
\begin{equation}
L_\alpha f(u)=\sum_{e=(u,v)\in E} \alpha(e)(f(u)-f(v)),
\label{E:L}
\end{equation}
where $u, v\in V$ and $\alpha(e)$ is the weight of edge $e$.
When this does not lead to confusion, we will use the notation $L$ instead of $L_\alpha$.

For each $k=(k_1,k_2)$ from the Brillouin zone $B=[-\pi,\pi)^2$,
let $L(k)$ be the \textbf{Bloch Laplacian}. Since there are two vertices inside the fundamental domain, each  operator $L(k)$ (or $L(z)$ in the multiplier notations)  acts on a two-dimensional space of
functions defined on the two atoms, so
%
$$
\sigma(L)=\mathop{\bigcup}\limits_{k\in B} \sigma(L(k))=\mathop{\bigcup}\limits_{k\in B}\{\lambda_1(k),\lambda_2(k)\},
$$
where $\lambda_1(k)\leq\lambda_2(k)$.

Our second main result is the following.

\begin{theorem}
  The dispersion relation of the operator $L_\alpha$ generically
  (i.e., outside of an algebraic subset of the parameters $\alpha$) satisfies all three conditions of
  Conjecture~\ref{Con:gener}.
\label{T:example}
\end{theorem}

\section{Proof of Theorem \ref{T:example}}

For the reasons we have described before, we provide three approaches to establishing Theorem~\ref{T:example}, each of which has its advantages and disadvantages and thus could be useful in different situations.
The first uses the paradigm of numerical algebraic geometry~\cite{SW05}.
While it yields a detailed understanding of the irreducible component structure of the set
$$DC\subset\C^9_\alpha\times(\C\backslash\{0\})^2\times\C_\lambda
$$
of degenerate critical points of the dispersion relation, it is not a traditional proof as the results of the numerical
computation are not certified in the sense of~\cite{HS12}. Besides, this is an extremely laborious computation, clearly unfeasible for more general graphs and higher-dimensional periodicity.
The second argument uses Theorem~\ref{T:dichotomy}, and it is probabilistic in the sense of Corollary~\ref{C:test}.
We give a third argument which is a proof in the traditional sense
and is based on intricate algebraic geometry results.

\subsection{Analytic reduction}\label{SS:reduction}
We spare the reader from the straightforward explicit computation of the symbol $L_\alpha(z)$, where the nine-dimensional vector $\alpha$ contains the weight parameter (``metric'') of the graph. It is a $2\times 2$ matrix Laurent polynomial in $z$.

The dispersion relation is
 \begin{equation}
  \label{E:1st}
    \lambda^2- \lambda \textrm{Tr}L_\alpha(z) + \det L_\alpha(z)=0.
 \end{equation}
The critical point condition gives
 \begin{equation}
   \label{E:2nd}
   \lambda \frac{\partial(\textrm{Tr}A)}{\partial z_j} - \frac{\partial(\det A)}{\partial z_j}=0
    \mbox{ for }j=1,2.
 \end{equation}
Finally, one can calculate that %
the vanishing of the Hessian determinant implies that
 \begin{multline}\label{E:3rd}
   \quad \left(\lambda \frac{\partial^2 (\textrm{Tr}A)}{\partial z_1^2}
     -\frac{\partial^2(\det A)}{\partial z_1^2}\right)\cdot
    \left(\lambda \frac{\partial^2 (\textrm{Tr}A)}{\partial z_2^2}
     -\frac{\partial^2(\det A)}{\partial z_2^2}\right)\\
    -\ \left(\lambda \frac{\partial^2 (\textrm{Tr}A)}{\partial z_1 \partial z_2}
       -\frac{\partial^2(\det A)}{\partial z_1 \partial z_2}\right)^2=0.\qquad
 \end{multline}

Let $g_1,\dotsc,g_4$ be, respectively, the left hand sides of the characteristic equation~\eqref{E:1st}, the two
equations~\eqref{E:2nd} for critical points of the dispersion relation, and the Hessian equation~\eqref{E:3rd}.
These are rational functions in the variables
$\alpha_1, \ldots, \alpha_9, z_1, z_2,\lambda $ with an interesting structure.
They are polynomials of degrees 2, 2, 2, and 4 in $\alpha_1, \ldots, \alpha_9, \lambda$ (homogeneous in $\alpha$) and
Laurent polynomials in $z_1,z_2$---their denominators all have the form $z_1^{n_1} z_2^{n_2}$ for some $n_1, n_2\in\N$.
Thus $g_1,\dotsc,g_4$ are Laurent polynomials that are defined on $\C^9_\alpha\times(\C\backslash\{0\})^2\times\C_\lambda$.
Let $DC$ be the algebraic variety defined by $g_1=g_2=g_3=g_4=0$.
This implies our first lemma.

\begin{lemma}\label{L:(1)}
  The set $DC_\R$ of points on the dispersion relation for the family $L_\alpha$ having degenerate critical points is a
  subset of the set of real points of the algebraic variety $DC$.
\end{lemma}

To study the variety $DC\subset \C^9\times(\C\backslash\{0\})^2\times\C$, for each $j=1,\dotsc,4$, let $f_j$ be the
numerator of $g_j$.
Then $f_1,\dotsc,f_4$ are ordinary polynomials in $\alpha,z,\lambda$.
Let $P\subset\C^{12}$ be the algebraic variety defined by the vanishing of $f_1,\dotsc,f_4$.
Then $DC=P\cap  (\C^9\times(\C\backslash\{0\})^2\times\C)$, but we may have $DC\neq P$, as $P$ may have components where
$z_1z_2=0$.
Indeed that is the case.

\begin{lemma}\label{L:(2)}
  The dimension of $P$ is nine and the dimension of $DC$ is eight.
\end{lemma}

In Subsection~\ref{SS:Bertini} we describe the decomposition of $P$ into irreducible components, which proves
Lemma~\ref{L:(2)}.

The dimension of the image of an algebraic variety $X$ under a map is contained in an algebraic variety whose dimension is
at most that of $X$~\cite[Sect.~I.6.3]{Shaf}.
Combined with Lemma~\ref{L:(2)}, this implies our third lemma:
\begin{lemma}\label{L:(3)}
  The image $DV$ of $DC$ under the projection to $\C^9_\alpha$ has dimension at most eight.
\end{lemma}
This completes the proof of Theorem~\ref{T:example}.
$\square$

\subsection{Numerical algebraic geometry verification of Lemma~\ref{L:(2)}.}\label{SS:Bertini}

We describe computations that establish Lemma~\ref{L:(2)} and therefore Theorem~\ref{T:example}.
They, as well as the derivations of Subsection~\ref{SS:reduction}, are archived on the website that accompanies this
article\footnote{{www.math.tamu.edu/\~{}sottile/research/stories/dispersion/}}.

We used the software Bertini~\cite{Bates_Bertini}, which is freely available and implements many
algorithms in numerical algebraic geometry~\cite{SW05}.
We start with the polynomials $f_1,\dotsc,f_4$, which are the numerators of~\eqref{E:1st}, \eqref{E:2nd}, and
\eqref{E:3rd} and which define the algebraic variety $P\subset \C^{12}$.
Bertini used the algorithms of regeneration~\cite{regeneration}, numerical irreducible decomposition~\cite{NID}, and
deflation~\cite{deflation} to study $P$, determining its decomposition into irreducible components, as well as
the dimension and degree of each component, and 
the multiplicity of $P$ along that component.
We sketch the consequence of that computation.

First, while $P$ is defined by four equations in $\C^{12}$, one finds that it has ten irreducible components of dimensions eight
and nine.
Specifically, it has seven components of dimension eight and three of dimension nine.
On all three components of dimension nine one of $z_1$ or $z_2$ vanishes.%
These components do not lie in $DC$ as $z_1z_2\neq 0$ on $DC$.
The consequence is that $DC$ has dimension eight, and therefore implies Lemma~\ref{L:(2)} and thus Theorem~\ref{T:example}.


This computation does not constitute a traditional proof, as Bertini does not certify its output.
We give an alternative verification using the dichotomy of Theorem~\ref{T:dichotomy} that relies on a symbolic
computation, and then a proof that uses geometric combinatorics.

\subsection{Symbolic computation verification}\label{SS:symbolic}
We describe how one can generate an ``almost surely'' verification of Theorem~\ref{T:example}, using
symbolic computation, which is exact and certified.
For this, we select a ``random'' (e.g., using a random numbers generator) point $\alpha\in \R_+^9$ and check that $\alpha\not\in DV$ and hence that
$\alpha\not\in DV_\R$ by showing that $DC\cap(\{\alpha\}\times(\C\backslash\{0\})^2\times\C)$ is empty. An application of Theorem \ref{T:dichotomy} then shows that Theorem~\ref{T:example} is almost surely valid,
in the sense of Corollary~\ref{C:test}. In presence of a true random numbers generator, this would prove the statement ``almost surely.''

We have done this for our graph as follows. First of all, due to a homogeneity present in the equation, instead of approximating the uniform distribution of $\alpha_j$ on a dense grid in a segment, say $[0,1]$, the equivalent option is to chose $\alpha$s uniformly among a large segment of natural numbers\footnote{Another argument is that the integers are Zariski-dense in the set of complex numbers.}. We thus picked each $\alpha_j$ randomly among the integers $\{1, 2, \dots, 50\}$. The random choice was $[31, 1, 13, 19, 36, 4, 27, 3, 7]$.

\begin{example}\label{L:empty}
  Let $\alpha=(31, 1, 13, 19, 36, 4, 27, 3, 7)$.
  Then there are no points $(z,\lambda)\in (\C\backslash\{0\})^2\times\C$ such that
  $g_1,\dotsc,g_4$ all vanish at $(\alpha,z,\lambda)$.
\end{example}

\begin{proof}
  A complex number $z$ is non-zero $(z\in\C\backslash\{0\})$ if and only if there exists $u\in\C$ with $zu=1$.
  For $\alpha=(31, 1, 13, 19, 36, 4, 27, 3, 7)$, a  Gr\"obner basis computation in both Maple and Singular~\cite{Singular} shows that
  the ideal $I$ in $\Q[\lambda,z_1,z_2,u_1,u_2]$ generated by $f_1,\dotsc,f_4, z_1u_1-1, z_2u_2-1$ contains $1$.
  But then $I$ defines the empty set, by Hilbert's Nullstellensatz.
\end{proof}

This (very fast) computation was performed for 100 different random choices of a parameter point $\alpha$, each time providing the same result.
One can claim with high confidence that the conjecture holds for this graph.

Fig.~\ref{Fig:Disp} shows another example: the dispersion relation for $L_{(1,2,3,4,5,6,7,8,1)}$. The horizontal plane is at $\lambda=0$ and the domain is $k\in[-\frac{\pi}{2}, \frac{3}{2}\pi]$.

\begin{figure}[ht!]
  \centering
  \includegraphics{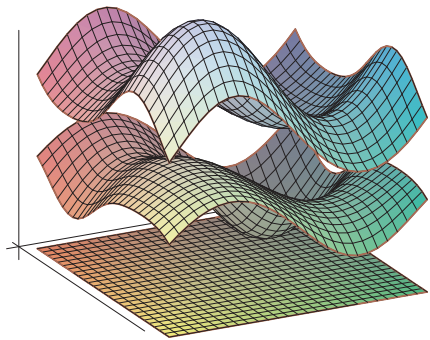}
  \caption{Dispersion relation for $L_{(1,2,3,4,5,6,7,8,1)}$.}
  \label{Fig:Disp}
\end{figure}

\subsection{Combinatorial algebraic geometry proof}\label{S:ActualProof}
We now present a valid proof of Theorem~\ref{T:example}.
Since $A_\alpha(z)$ is a $2\times 2$ matrix, its characteristic polynomial~\eqref{E:1st}
is quadratic in $\lambda$ with leading coefficient 1, and thus the variety it defines in
$\C^9\times(\C\backslash\{0\})^2\times\C$ (the dispersion relation) has the property that its
projection to the $(\alpha,z)$ parameters $\C^9\times(\C\backslash\{0\})^2$ is a proper map
with each fiber consisting of either two points, or a single point of multiplicity 2.

Consider now the set $CP$ which is defined by the characteristic equation~\eqref{E:1st} and
the two critical point equations~\eqref{E:2nd}.
Then $CP$ consists of points $(\alpha,z,\lambda)$ such that $\lambda$ is a critical point of the dispersion relation.
The critical points $(z,\lambda)\in(\C\backslash\{0\})\times\C$ for any given $\alpha\in\C^9$ are the set of solutions to
three equations in the three variables $z_1,z_2,\lambda$.
A celebrated result of Bernstein~\cite{Bernstein} gives a strict upper bound
for the number of isolated solutions, counted with multiplicity.

Let us explain this.
The exponent of a monomial $z_1^{a_1}z_2^{a_2}\lambda^{a_3}$ is an integer vector $(a_1,a_2,a_3)\in\Z^3$.
The exponents that occur in the nonzero terms in a Laurent polynomial $f$ form its support.
Their convex hull $N(f)$ is the Newton polytope of $f$.
The bound in Bernstein's theorem is given by Minkowski's mixed volume of the Newton polytopes of the equations.

The polynomials $f_1,f_2,f_3$ which define $CP$ have the following Newton polytopes.
 \begin{equation}\label{Eq:NewtonPolytopes}
  \raisebox{-40pt}{\includegraphics[height=88pt]{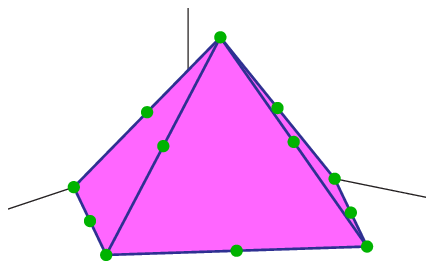}\quad
                   \includegraphics[height=70pt]{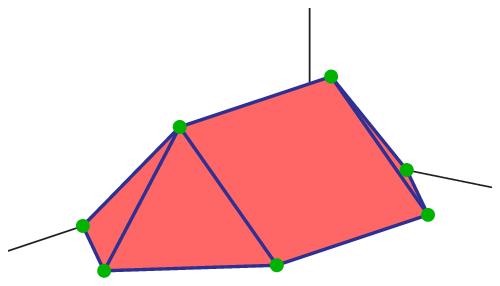}\quad
                   \includegraphics[height=70pt]{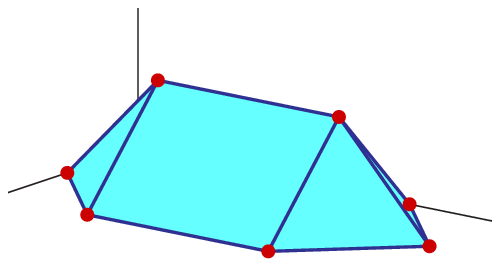}}
 \end{equation}
For the characteristic equation $f_1$, this is the pyramid with vertex $(2,2,2)$ and base the
square with vertices $(0,2,0)$, $(2,0,0)$, $(4,2,0)$, and $(2,4,0)$.
The side length of each edge in the base is $2\sqrt{2}$ and its height is 2, so that its volume is $32/6$.
The other two are reflections of each other.
The first has base the hexagon with vertices
\[
  (0,1,0),\ (1,0,0),\ (3,0,0),\ (4,1,0),\ (3,2,0),\ (1,2,0),
\]
and its other vertices are $(1,1,1)$ and $(3,1,1)$.
If all polytopes are translated so that the centers of their bases are at the origin
$(0,0,0)$, then the second two lie inside the pyramid.

\begin{lemma}\label{L:MV}
  The mixed volume of the three polytopes~\eqref{Eq:NewtonPolytopes} is $32$.
\end{lemma}

\begin{proof}
 This is a consequence of a result of Rojas~\cite[Cor.~9]{Rojas}, which is explained in~\cite[Cor.~3.7]{BiSo}.
 Let the three translated polytopes be $P$, $Q$, and $R$, with $P$ being the pyramid.
 Then $P=P\cup Q\cup R$.

 Observe that at least one of the three polytopes ($P$) meets every vertex of $P$.
 Also, for every edge $e$ of $P$, at least two have  an edge lying along $e$.
 Finally, for every facet $F$ of $P$, all three polytopes have a facet lying along $F$.
 A consequence of~\cite[Cor.~9]{Rojas} (an explanation in \cite{BiSo} is useful) is that the mixed volume of $P,Q,R$ is $3!=6$ times
 the volume of $P$, which is 32.
\end{proof}

\begin{proof}[Proof of Theorem~\ref{T:example}]
  For the point $\alpha=(1,2,3,4,5,6,7,8,1)$, we use Maple and Singular to show that there
  are 32 nondegenerate critical points of the dispersion relation.
  By Lemma~\ref{L:MV}, this is the maximal number of critical points, and so we conclude
  that $\alpha$ is a regular value of the projection map $\pi\colon CP\to\C^9$.
  Furthermore, there is a neighborhood $U$ of $\alpha$ such that over it the map $\pi$ is a 32-sheeted cover and therefore is
  proper near $\alpha$.

  The set $DC$ of degenerate critical points is a closed subset of $CP$ and thus its projection to $\C^9$ is proper near
  $\alpha$.
  Thus its image $DV$ is closed in a neighborhood of $\alpha$.
  Since $\alpha\not\in DV$, this implies that the complement of $DV$ in $\C^9$ contains a neighborhood of
  $\alpha$.
  But any nonempty classical open subset of $\C^9$ is Zariski-dense, and therefore we conclude that
  the complement of $DV$ in $\C^9$ contains a nonempty Zariski open set, which completes the proof.
\end{proof}

An interesting outcome from this version of the proof is the following result:
\begin{theorem}\label{T:32}
Under the conditions of Theorem \ref{T:example}, the statement of Conjecture \ref{C:upscaling} holds true. In other words, increasing the number of parameters in (e.g., adding more edges to) the two-atomic structure shown in Fig. \ref{F:sample} does not change the conclusion on genericity of Theorem \ref{T:example}.
\end{theorem}
Indeed, the crucial mixed-volume computation of Lemma \ref{L:MV} does not react to increasing the number of parameters $\alpha$.

\section{Degenerate  subgraphs}\label{S:subgraphs}

Lemma~\ref{L:empty} and Corollary~\ref{C:test} provide an efficient method to study Conjecture~\ref{Con:gener} on
sufficiently simple discrete periodic graphs.
We illustrate this on a case study involving all $2^9$ subgraphs of the graph of Fig.~\ref{F:sample}, corresponding to
choosing a subset $S$ of the nine edges.

Let $f_1,\dotsc,f_4\in\Q[\alpha,\lambda,z_1,z_2]$  be the polynomials that define the variety $P$ following
Lemma~\ref{L:(2)}.
Given a subset $S$ of the nine edges, let $I_S\subset\Q[\alpha,\lambda,z_1,z_2,u_1,u_2]$ be the ideal generated by
$z_1u_1-1$, $z_2u_2-1$, and the polynomials obtained from  $f_1,\dotsc,f_4$ by setting all parameters $\alpha_j$ equal to
zero for $j\not\in S$.
Then for parameters $\alpha_S=(\alpha_i\mid i\in S)$, $I_S$ vanishes on the set $DC$ of degenerate critical points on the
dispersion relation for the graph $\Gamma_S$ corresponding to $S$ with parameters $\alpha_S$.

We have a Maple script that, for each subset $S$, evaluates the ideal $I_S$ at ten random instances of the parameters
$\alpha_S$.
If, for each of these instances of the parameters $\alpha_S$ it finds that $1\not\in I_S$ (so that the corresponding
dispersion relation has a degenerate critical point), then it adds $S$ to a set $\DSG$ of degenerate subgraphs.
This set $\DSG$ contains 87 subsets $S$.
This set, according to Theorem \ref{T:32} has the structure of a simplicial complex.
That is, if $S\in\DSG$ and $T\subset S$, then $T\in\DSG$.
There are eleven maximal subsets (simplices) $S$ in $\DSG$, corresponding to eleven maximal subgraphs of the graph in
Fig.~\ref{F:sample} which always have degenerate dispersion relations.
We display them in Fig.~\ref{F:DSG}.
\begin{figure}[htb]
\centering
\includegraphics[height=105pt]{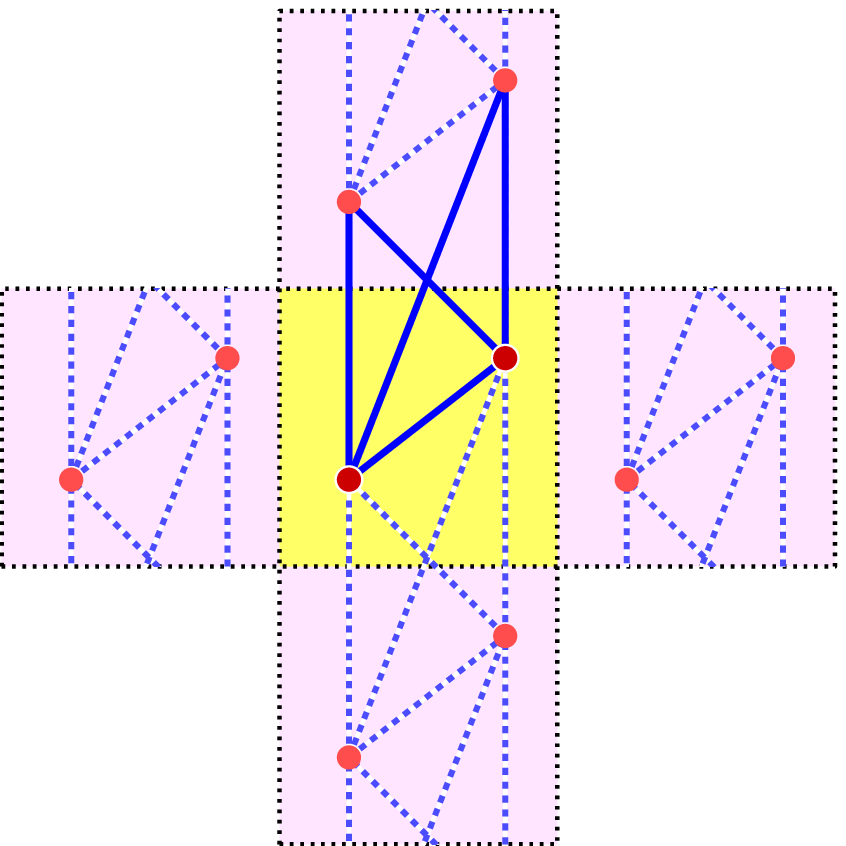}\quad%
\includegraphics[height=105pt]{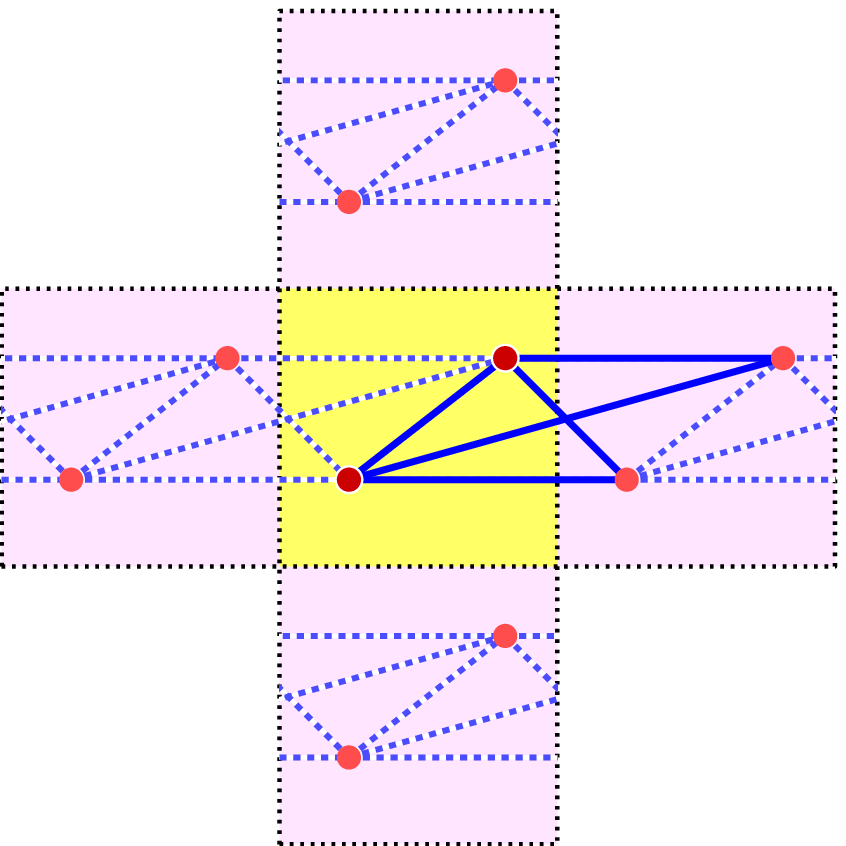}\quad%
\includegraphics[height=105pt]{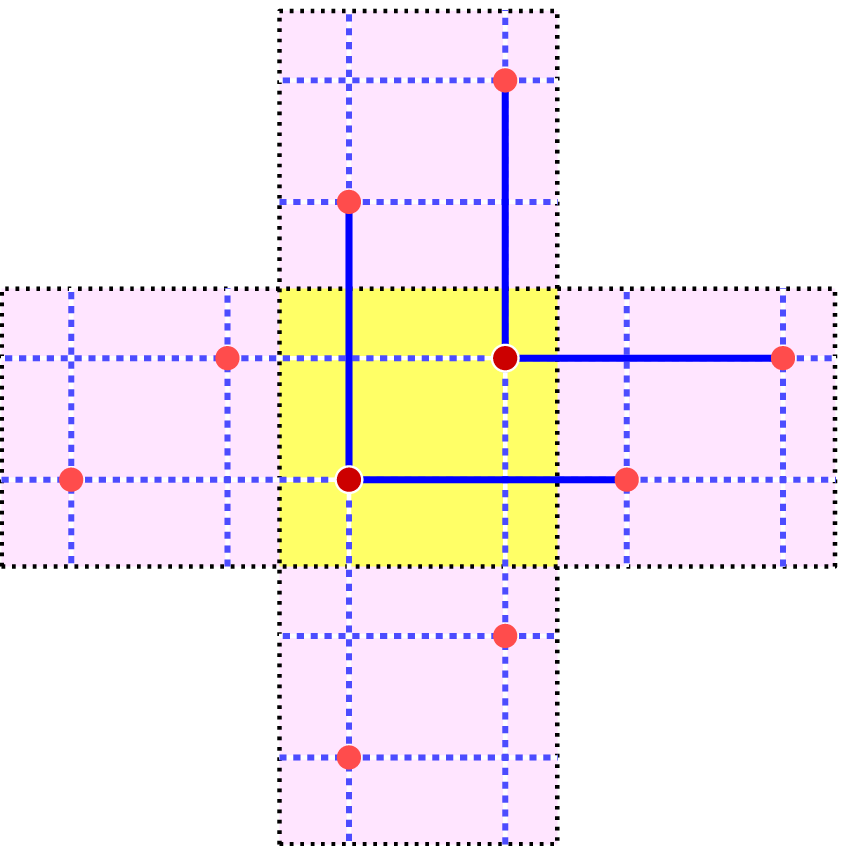}\bigskip

\includegraphics[height=105pt]{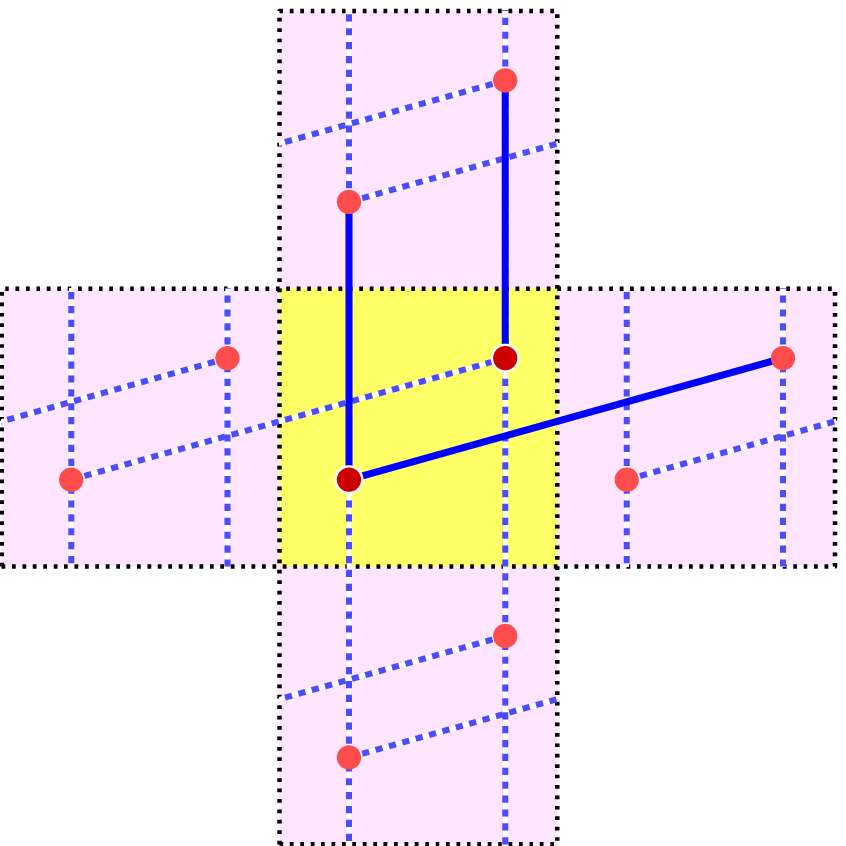}\quad%
\includegraphics[height=105pt]{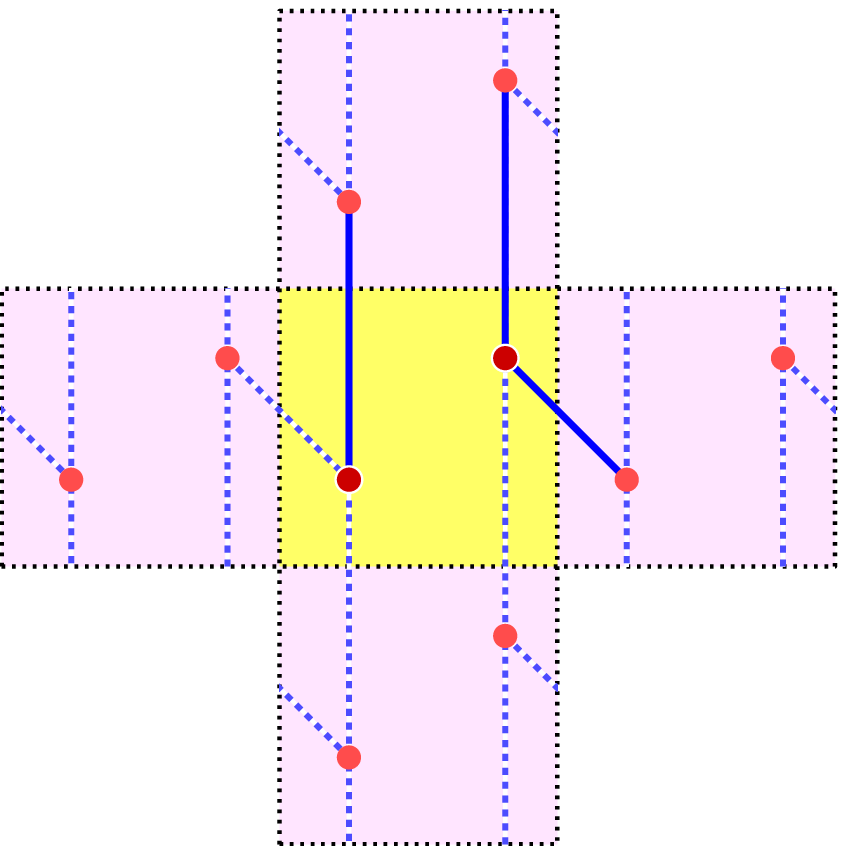}\quad%
\includegraphics[height=105pt]{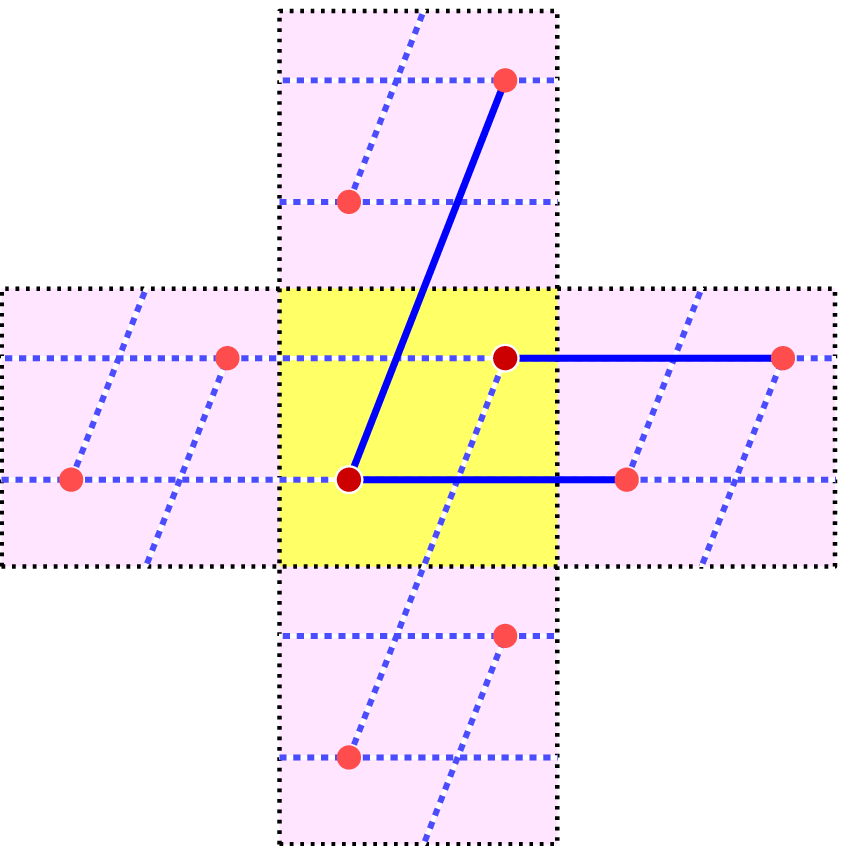}\bigskip

\includegraphics[height=105pt]{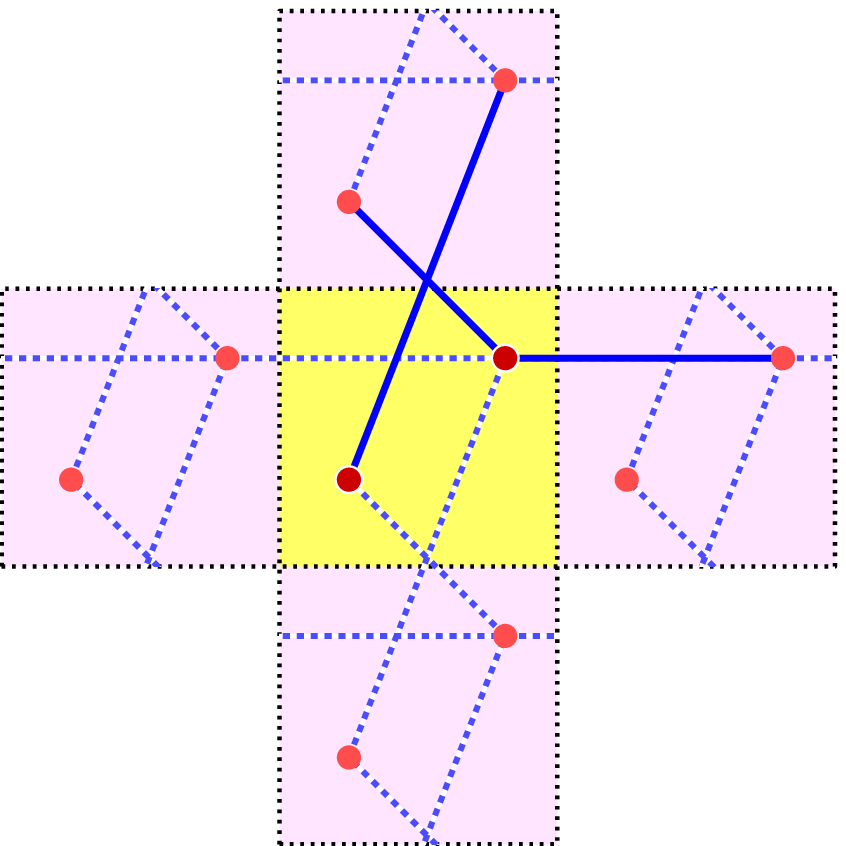}\quad%
\includegraphics[height=105pt]{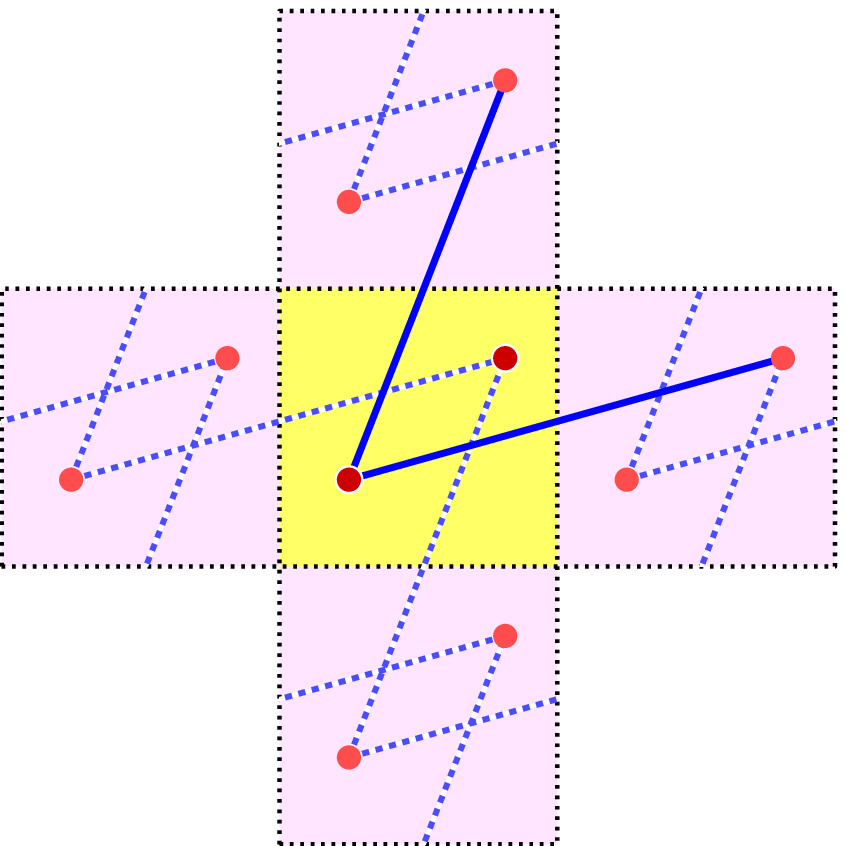}\quad%
\includegraphics[height=105pt]{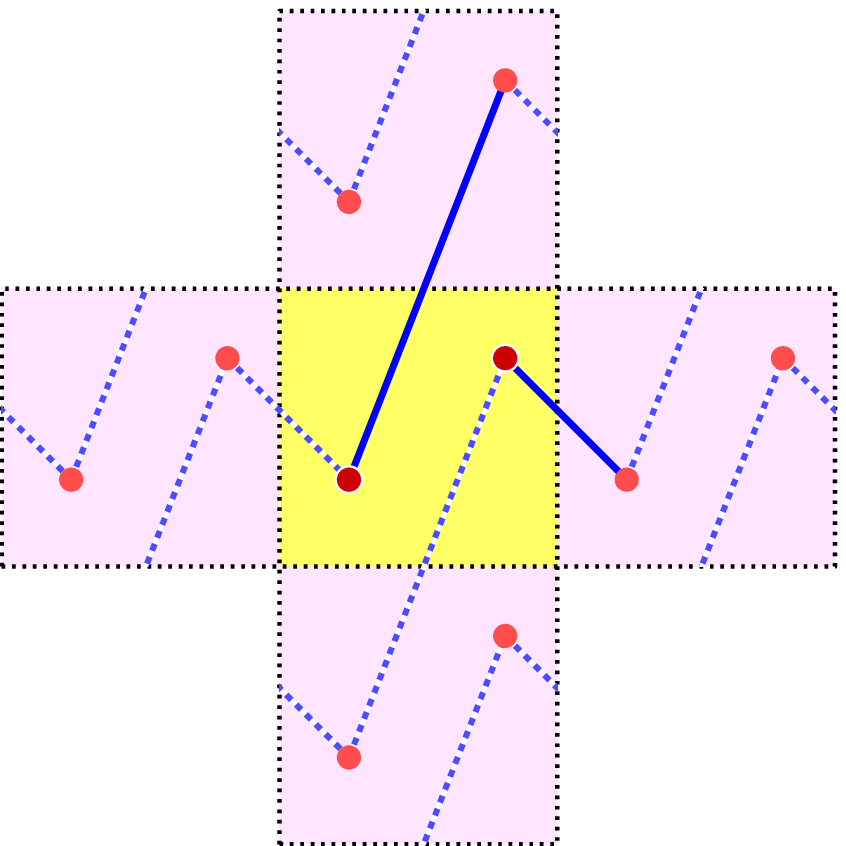}

\includegraphics[height=105pt]{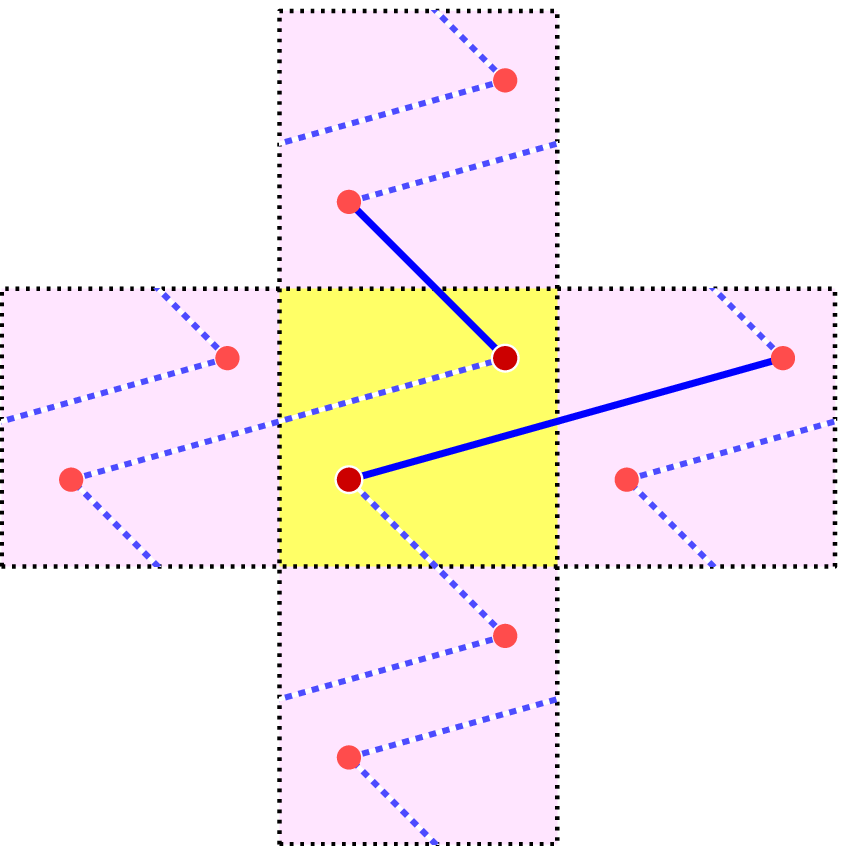}\quad\includegraphics[height=105pt]{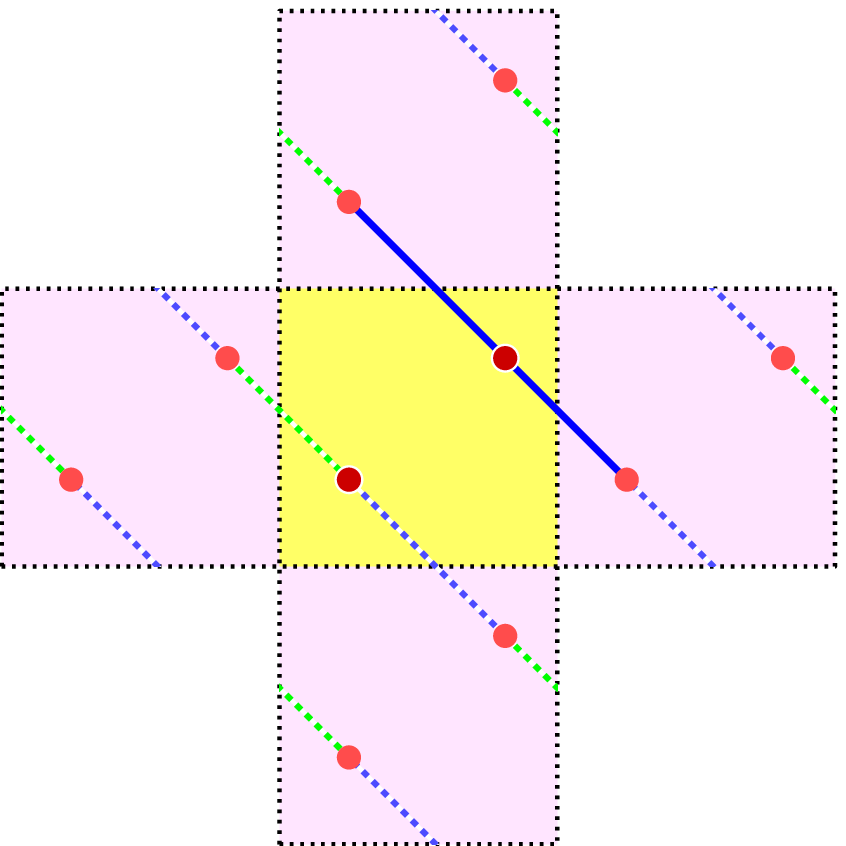}

\caption{Maximal degenerate subgraphs of the graph of Fig.~\ref{F:sample}.}
\label{F:DSG}
\end{figure}
Observe that each of these graphs are disconnected.
Moreover, they consist either of one or more $\Z$-periodic graphs together with their disjoint copies under translation by
$\Z$, (thus providing jointly a $\Z^2$-periodicity), or two disjoint isomorphic copies of a $\Z^2$-periodic graph.
It is a simple exercise that both situations lead to degeneration (even in the continuous case).
Thus all degenerations occur for ``obvious'' reasons only.


\section{Conclusions and final remarks}\label{S:remarks}
\begin{enumerate}
\item We show a dichotomy of the degeneracy question for general algebraically fibered matrix operators. It follows from a simple algebraic geometry argument.

We show that, in spite of existing counterexamples in general, the generic non-degeneracy conjecture does hold for the ``mother-graph'' (from Fig. \ref{F:sample}) of $\Z^2$-periodic two-atomic structures with nearest cell interaction.

\item All subgraphs of this graph for which genericity does not hold were found. The degeneracy appears there only for trivial reasons.

    There are 98 disconnected subgraphs of the graph of Fig.~\ref{F:sample}, which also
  form a simplicial complex.
  Of those that do not appear in Fig.~\ref{F:DSG}, we show the three which are maximal in
  Fig.~\ref{Fdisconnect}.%

\item
  The initial impression is (see Conjecture \ref{C:upscaling} and its confirmation in our example in Theorem \ref{T:32}) that one needs to have sufficiently many free
parameters in the operator to expect generic non-degeneracy, and that adding more parameters does not destroy genericity.
It would be, however, interesting to have better understanding
of what makes some discrete periodic problems degenerate.
The examples of Section~\ref{S:subgraphs} may be instructive.

\item
  It may be possible to extend our analysis in Section~\ref{S:ActualProof} to other periodic graphs.
A starting point could be to determine the Newton polytopes and mixed volumes of the polytopes corresponding to the
equations defining the set $CP$ of critical points of the dispersion relation.

We provide three approaches, which might be useful in different situation: a (non-certified) computational algebaric one, which uses Bertini software to find all irreducible components and their dimensions. The second one is testing a random point in the parameter space, which provides the correct answer almost surely. It uses Gr\"obner basis technique and can be run fast. The third one uses mixed volumes of polytopes and Bernstein-Kushnirenko type results to prove the generic non-degeneracy.
\item It is easy to create (see \cite{BerKuc}), using non-trivial graph topology, compactly supported eigenfunctions. This is known to lead to appearance of flat components in the Bloch variety, and thus degenerate extrema. However, this situation is non-generic: it is destroyed by generic small variations of the lengths (weights) of edges.
\item The reader should notice that we quickly abandon discussion of the spectral edges only and target a loftier goal - all critical points. It might be easier to understand the generic structures of (much fewer) spectral edges, but the authors have not figured out how to use this distinction.
\item
In the continuous case, the dispersion relations are not algebraic, and the operators $L(z)$ are unbounded and thus the
projection of the set $DC$ is NOT a proper map.
One thus needs to deal with projections of sets of zeros of entire functions into subspaces, such as for instance in the
classical theorem by Julia \cite{Julia}.
The situation there is complicated, and to get any reasonable results about projections of such analytic sets, one needs
more information, e.g.\ assumptions on the growth of the defining function \cite{Alexander,Erem}.
Although such growth estimates do exist (see, e.g.~\cite{KuchBAMS,KuchBook}), the authors have not succeeded in establishing
similar results for, say, Schr\"odinger operators with periodic potentials.
We, however, conjecture that an analog of the dichotomy Theorem \ref{T:dichotomy} should hold, with genericity understood
in the (unavoidably weaker) Baire category sense.
\end{enumerate}

\begin{figure}[htb]
\centering
  \includegraphics[height=105pt]{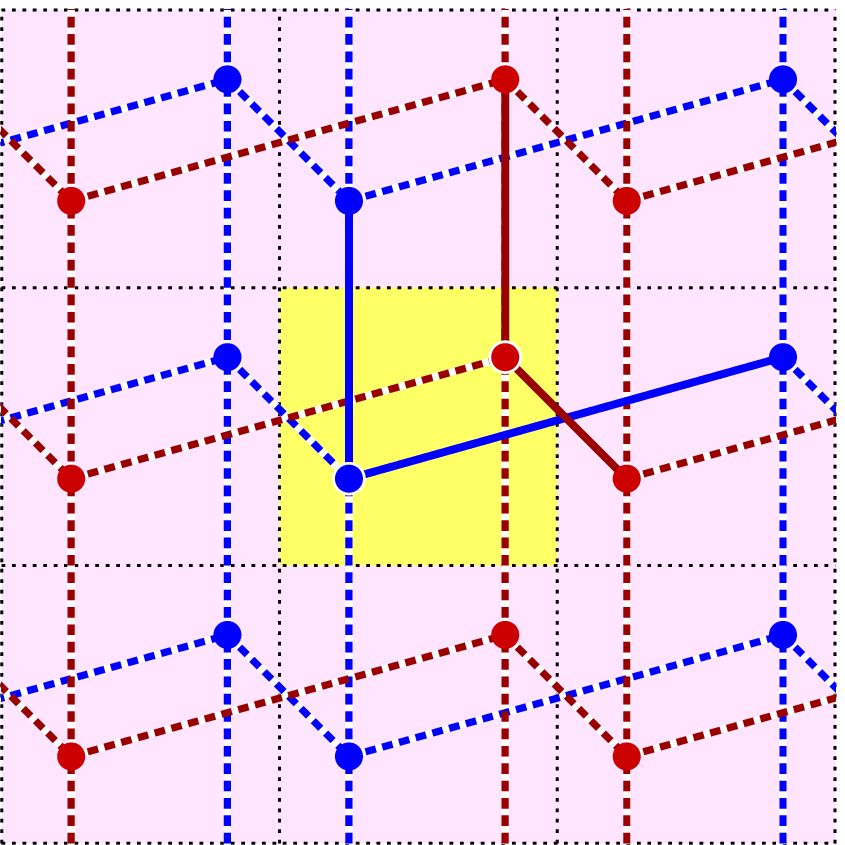}\quad%
  \includegraphics[height=105pt]{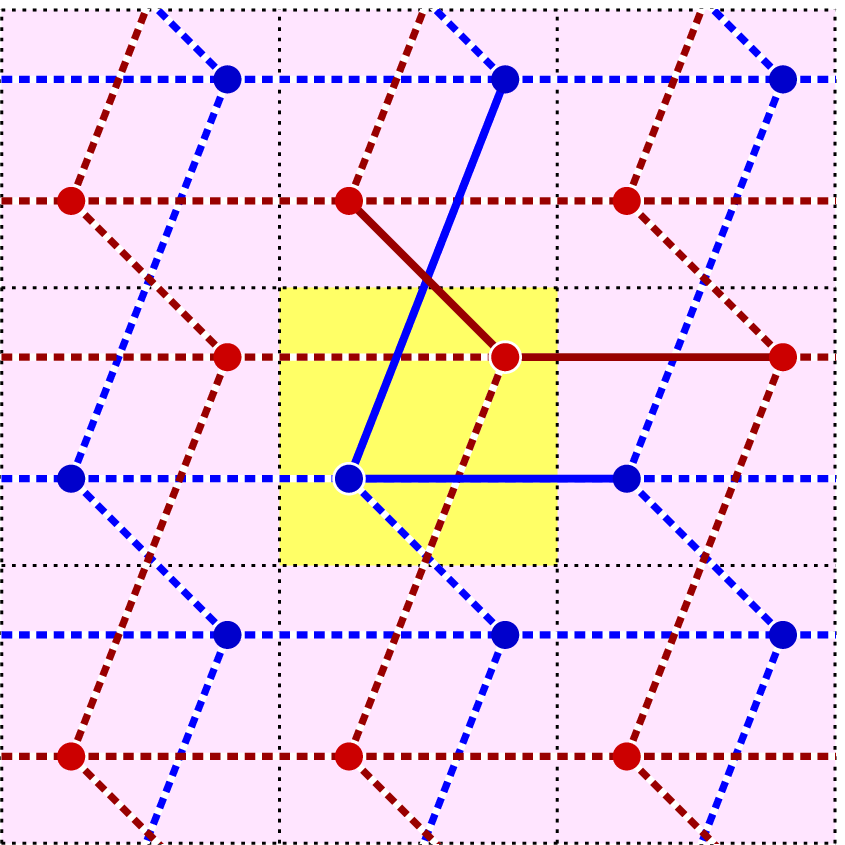}\quad%
  \includegraphics[height=105pt]{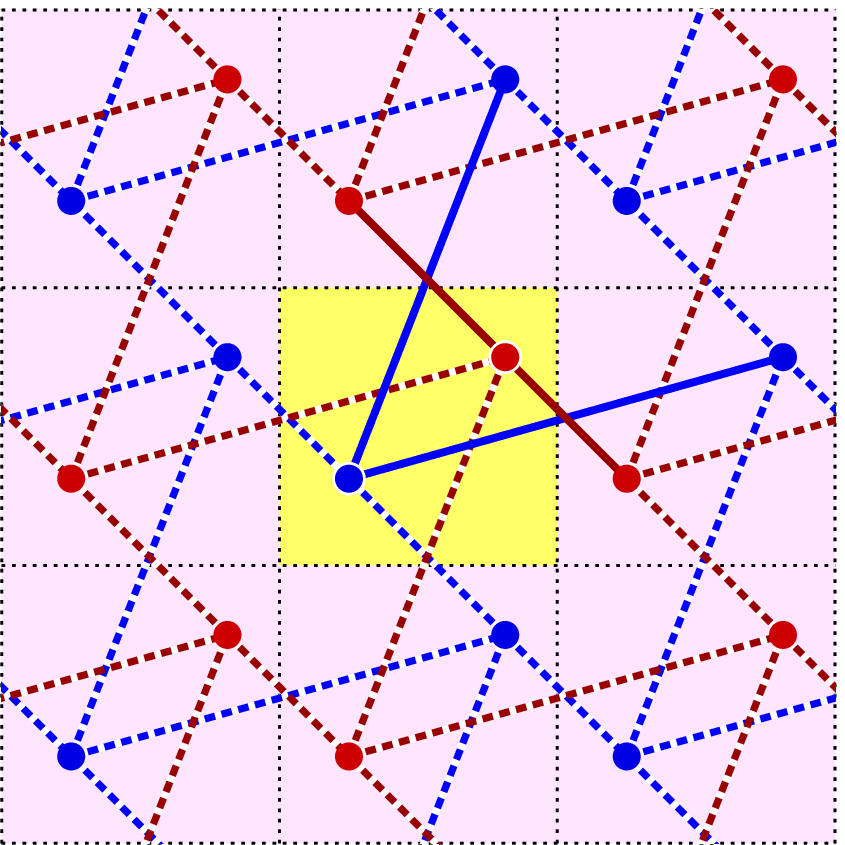}
\caption{Maximal disconnected subgraphs not in $\DSG$.}
\label{Fdisconnect}
\end{figure}

\section{Acknowledgments}\label{S:thanks}
The authors acknowledge support provided by the NSF. We are also grateful to G.~Berkolaiko, N.~Filonov, J.~Hauenstein,
I.~Kachkovskiy, Minh~Kha, L.~Parnovsky, and R.~Shterenberg for discussions and information. Thanks also go to the reviewers, whose comments have lead to significant improvement of the text.


\end{document}